\documentclass[letterpaper, 10 pt, conference]{ieeeconf}

\IEEEoverridecommandlockouts                              

\usepackage[utf8]{inputenc}
\usepackage{include_packages} 
\usepackage[font=footnotesize,skip=5pt]{caption}
\usepackage[english]{babel}

\title{Bridging the Gap Between Central and Local Decision-Making:\\
The Efficacy of Collaborative Equilibria in Altruistic Congestion Games}
\author{{Bryce L. Ferguson, Dario Paccagnan, Bary S. R. Pradelski, and Jason R. Marden}
\thanks{This research was supported by ONR Grant \#N00014-20-1-2359, AFOSR Grant \#FA95550-20-1-0054 and \#FA9550-21-1-0203}%
\thanks{B. L. Ferguson (corresponding author) and J. R. Marden are with the Department of Electrical and Computer Engineering, University of California, Santa Barbara, CA, {\texttt{\{blferguson,jrmarden\}@ece.ucsb.edu}}.}%
\thanks{D. Paccagnan is with the Department of Computing, Imperial College London, {\texttt{\{d.paccagnan\}@imperial.ac.uk}}.}%
\thanks{B. S. R. Pradelski  is with the National Center for Scientific Research (CNRS), {\texttt{\{bary.pradelski@cnrs.fr\}@cnrs.fr}}.}%
}

\begin{document}
\maketitle 

\begin{abstract}
Congestion games are popular models often used to study the system-level inefficiencies caused by selfish agents, typically measured by the price of anarchy.
One may expect that aligning the agents' preferences with the system-level objective--altruistic behavior--would improve efficiency, but recent works have shown that altruism can lead to more significant inefficiency than selfishness in congestion games.
In this work, we study to what extent the localness of decision-making causes inefficiency by considering collaborative decision-making paradigms that exist between centralized and distributed in altruistic congestion games.
In altruistic congestion games with convex latency functions, the system cost is a super-modular function over the player's joint actions, and the Nash equilibria of the game are local optima in the neighborhood of unilateral deviations.
When agents can collaborate, we can exploit the common-interest structure to consider equilibria with stronger local optimality guarantees in the system objective, e.g., if groups of $k$ agents can collaboratively minimize the system cost, the system equilibria are the local optima over $k$-lateral deviations.
Our main contributions are in constructing tractable linear programs that provide bounds on the price of anarchy of collaborative equilibria in altruistic congestion games.
Our findings bridge the gap between the known efficiency guarantees of centralized and distributed decision-making paradigms while also providing insights into the benefit of inter-agent collaboration in multi-agent systems.
\end{abstract}


\section{Introduction}
In many multi-agent systems, each agent's self-interested decisions drive the overall system behavior. Examples include individual vehicles in transportation networks, producers and consumers in supply chains and power grids, or servers and other computational resources in cloud computing.
In each of these settings, the resulting system performance can be sub-optimal relative to what is feasibly attainable~\cite{Lazar2018,Perakis2007,epstein2012price}.
This inefficiency can primarily be attributed to the fact that agents 1) make decisions in their own self-interest (not in line with the system-level objective) and 2) make decisions locally without coordinating their behavior.
Many works have studied how to reduce this inefficiency by altering agents' objectives (e.g., with monetary incentives~\cite{Maille2012} or utility design~\cite{makabe2022utility}), but far fewer have closely studied the effect of localized decision-making has on this inefficiency.

In the context of atomic congestion games~\cite{Rosenthal1973}, it is well known that the Nash equilibria need not be optimal with respect to the system-level objective of minimizing total congestion~\cite{Christodoulou2005,Koutsoupias2009}.
Many works study how to quantify this inefficiency via the price of anarchy ratio~\cite{Roughgarden2012} and how to reduce it by altering agents' preferences using monetary~\cite{Paccagnan2021I,Bilo2016,Vijayalakshmi2020} or informational~\cite{Castiglioni2020} incentives.
Surprisingly, when agents' costs are directly aligned with the system objective, the price of anarchy can be greater than if agents were purely self-interested~\cite{caragiannisTaxesLinearAtomic2006,Paccagnan2021I}, that is, altruism can worsen inefficiency.
Further, when agents are restricted to making decisions locally, there are fundamental limits to our capabilities in reducing inefficiency~\cite{Paccagnan2021a}.
In a variety of engineered systems, new technologies enable us to think about paradigms where decisions are not fully local, such as ride-sharing platforms with human and autonomous vehicles~\cite{chau2022DecentralizedRideSharingVehiclePooling}, smart grids with active `prosumers'~\cite{bistaffa2012decentralised}, and fleets with robot-to-robot communication~\cite{das2019tarmac}.
To help understand the benefits that new communication technologies and market platforms can provide, we study the opportunities presented when going beyond local decision-making and allowing for partial coordination among agents in the system.

This work focuses on bridging the gap between centralized and distributed performance by studying \emph{collaborative decision-making paradigms}.
A collaborative decision-making environment is defined by: 1) which agents are enabled to collaborate and 2) how a group makes a decision.
The manner in which the agents of a multi-agent system can collaborate is context-dependent;
To range different levels of collaboration, we consider collaborative structures in which groups of up to $k$ agents can form to update their group strategy.
Among many possible forms of collaboration, we consider when members of a group share a common-interest objective, such that they cooperatively optimize a single objective function over their group actions; when the agent objective and system objective are aligned, we recover the notion of altruistic decision-making~\cite{chen2014AltruismItsImpact} and extend it to the collaborative setting.
By varying $k$ between 1 (fully distributed) and the total number of agents (fully centralized), we can sweep through different levels of collaboration.


\noindent\textit{Related Work - }
Coalitions in games have been studied in many contexts, including bargaining~\cite{tang2018balanced},
cost sharing~\cite{hoefer2013StrategicCooperationCost}, or team formation~\cite{bullinger2020ComputingDesirablePartitions}.
In each of these settings, agents join or stay in coalitions when it is favorable for the individual.
In this work, we consider collaborative coalitions where groups form not because they intrinsically prefer being in the group but because the group possesses a greater ability to alter the system behavior in the group's favor.
This is most closely related to the study of \textit{strong Nash equilibria}, in which groups of agents deviate their group action only if it is beneficial for each member of the group~\cite{aumann1959acceptable}.
Much work has gone into verifying the existence of $k$-strong Nash equilibria~\cite{nessah2014existence,harks2009strong} and providing methods to find them~\cite{holzman1997strong,gatti2017verification,clempner2020finding}.
Recent works have further studied the price of anarchy with respect to strong Nash equilibria in cost-sharing games~\cite{epsteinStrongEquilibriumCost2007}, smooth games~\cite{bachrachStrongPriceAnarchy2014}, job scheduling and resource allocation games~\cite{fiat2007strong,andelmanStrongPriceAnarchy2009,ferguson2023CollaborativeDecisionMakingKStrong}, and even congestion games~\cite{chien2009strong}.
However, these results are limited to more restricted problem classes or rely on the number of players being unbounded; further, many of these bounds are not tight, and in many settings, $k$-strong Nash equilibria need not exist.

In this work, we guarantee the existence of $k$-strong Nash equilibria by considering the special case of common-interest collaborations.
The authors' previous work~\cite{ferguson2023CollaborativeDecisionMakingKStrong} applied this approach to resource allocation problems; in this work, we apply similar techniques to altruistic congestion games.
In this work, we consider a new, more general model with a richer class of system objective functions and garner new insights on the effects of collaboration in an entirely different setting, connecting different areas of the literature.
Specifically, in congestion games, it has been shown that agents' altruistic behavior can actually increase the price of anarchy (i.e., worsen system performance); this finding points to the fact that inefficiency is caused by the localness of decision-making as much as the (mis)alignment of agent and system objectives.
Our results provide tools to understand the benefit of collaborative decision-making environments that go beyond local decision-making and help us to bridge the gap between centralized and distributed performance.
We provide tractable linear programs that provide bounds on the price of anarchy of $k$-strong Nash equilibria.
\cref{fig:spoa_grid} illustrates how these findings can be used to bridge the gap in performance guarantees of local and central decision-making in altruistic congestion games.



\section{Preliminaries}
Throughout, let $[n]:=\{1,\ldots,n\}$ and $\binom{n}{k}:=\frac{n!}{k!(n-k)!}$.

\subsection{Collaborative Decision-Making}\label{subsec:prelim_collab}
Consider a multi-agent system where $N = \{1,\ldots,n\}$ denotes a finite set of agents.
Each agent $i \in N$ has a set of available actions or strategies $\strat_i \in \Strats_i$.
When each agent has made a selection of their strategy, the group behavior is denoted by the joint strategy $\strat = (\strat_1,\ldots,\strat_n) \in \Strats = \Strats_1\times\cdots\times\Strats_n$.
The overall system's performance is dictated by the actions of each of the agents; let $C:\Strats \rightarrow \mathbb{R}_{\geq 0}$ denote the system cost function over joint strategies.
The system operator would like to induce a joint strategy that minimizes the system cost, i.e.,
$$\Opt{\strat} \in \argmin_{\strat \in \Strats} C(\strat).$$
Though these system states are ideal, it may be difficult for the system operator to attain them due to
the large scale and limited connectivity of the system,
an inability to directly control each agent's behavior,
or the possible intractability of computing an optimal joint strategy.
In light of this, we consider that agents make decisions in a decentralized manner and seek to quantify the system cost of the emerging joint strategies.

A candidate paradigm for decentralized decision-making is to distribute the strategy selection task to each local agent fully.
\sloppy Though this greatly reduces the communication and computational burdens, the joint strategy that agents settle on can be far from optimal~\cite{Koutsoupias2009}.
Further, with the advancement of new sensing and communication technologies, it's feasible that agents could form connections and cooperate in selecting their actions.
This could take the form of designed agents forming collaborative groups (e.g., self-driving cars entering platoons~\cite{ahangar2021survey}) or strategic agents colluding their behavior (e.g., human Taxi/Uber drivers influencing transit supply and pricing~\cite{chau2022DecentralizedRideSharingVehiclePooling}).
To understand the effects of coalitional decision-making on multi-agent system performance, this work studies the efficiency of equilibria that emerge from \emph{collaborative decision-making paradigms}.


A collaborative decision-making paradigm is defined by two parts: 1) which agents are enabled to collaborate, and 2) how a group of agents selects their strategies.
To the first point, let $\Combos \subseteq 2^N$ denote the collaboration set, dictating which groups of agents can collaborate on their actions.
In general, the form of this set will be context-dependent (e.g., agents communicating over a network or over local broadcasts), and the groups of agents enabled to collaborate will alter the behavior of the overall system: larger groups can have a bigger impact on the joint strategy, and agents existing in multiple groups can facilitate cascading effects of collaboration.

\begin{figure*}
    \centering
    \vspace{10pt}
    \includegraphics[width=0.9825\textwidth]{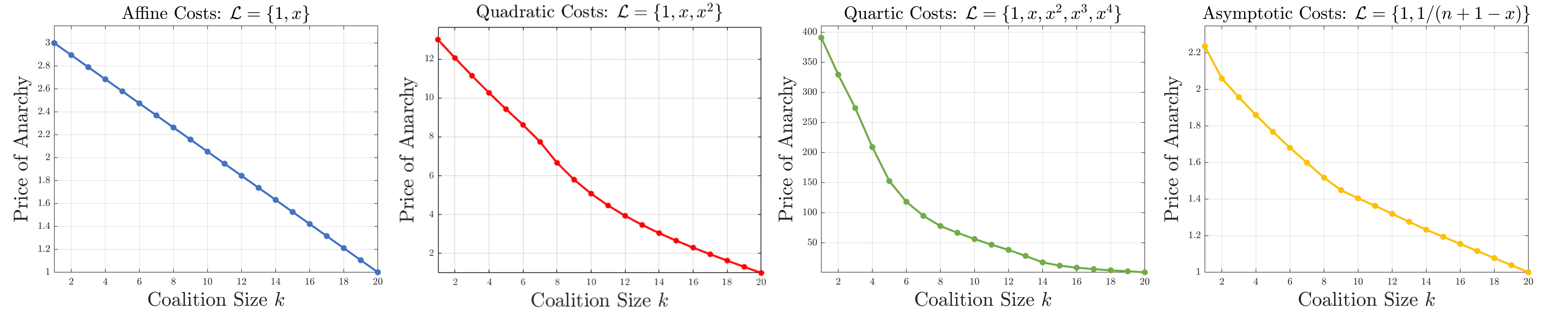}
    \caption{Price of Anarchy with different collaborative decision-making paradigms in various classes of congestion games. In each plot is the \kstrong price of anarchy within the respective class of altruistic congestion games with latency functions formed by linear combinations of the functions in $\Latset$. The bounds are generated by the linear programs in \cref{thm:alt}. In each setting, increasing the amount of collaboration ($k$) improves the equilibrium efficiency guarantee.}
    \label{fig:spoa_grid}
    \vspace{-10pt}
\end{figure*}

The other factor determining the behavior of a collaborative multi-agent system is how groups select actions.
We will consider that a group $\Gamma \in \Combos$ seeks to minimize their group cost function $J_\Gamma:\Strats \rightarrow \mathbb{R}_{\geq 0}$; that is, given some behavior of the remaining agents, $\strat_{-\Gamma}$, a best response for the group $\Gamma$ is to select an action in $\argmin_{\strat_\Gamma \in \Strats_\Gamma}J_\Gamma(\strat_\Gamma,\strat_{-\Gamma}).$
In this work, we consider \textit{common-interest collaboration}, in which each group $\Gamma$ minimizes a common objective function $\Phi$, i.e., $J_\Gamma(\strat) =\Phi(\strat)$ for all $\Gamma \in \Combos$.
At several points in this work, we will take a particular focus on the special case of common-interest collaboration where the agents' objectives are directly aligned with the system objective (i.e., $\Phi(\strat) = C(\strat)$), termed altruistic collaboration.
We will focus on the system behavior that emerges from such collaborative architectures by studying the stable states of collaborative decision-making, which can reveal the efficacy of system paradigms that go beyond distributed decision-making.

In a multi-agent system $G = (N,\Strats,C)$ with the common-interest collaborative decision-making structure $(\Combos,\Phi)$, a joint strategy $\strat^{\rm EQ}$ is an equilibrium if each group $\Gamma \in \Combos$ is simultaneously best responding to the current action, i.e.,
\begin{equation}\label{eq:def_gen_eq}
\Phi(\strat^{\rm EQ}) \leq \Phi(\strat_\Gamma^\prime,\strat_{-\Gamma}^{\rm EQ}),~\forall\strat_\Gamma^\prime \in \Strats_\Gamma,~\Gamma \in \Combos.
\end{equation}
That is, no group is incentivized to deviate from their group strategy in these states.
This definition is similar to that of the Nash equilibrium in game theory~\cite{nash1950equilibrium}, and indeed when $\Combos = \{i\}_{i \in N}$, the states satisfying \eqref{eq:def_gen_eq} are Nash equilibria of the common-interest game; further, if $\Phi$ is a potential function for some potential game, then the set of Nash equilibria and the set of states that satisfy \eqref{eq:def_gen_eq} are equivalent~\cite{Monderer1996}.
In this work, we seek to go beyond local decision-making and consider the effects of collaborative architectures $\mathcal{H}$ and how they alter the set of system equilibria.
A central observation of this work is that the collaborative equilibria that satisfy \eqref{eq:def_gen_eq} need not optimize the system-level objective $C$; we study the quality of collaborative equilibria to understand the change in system performance under different collaboration structures.

To quantify the efficiency of a collaborative design, we consider the \emph{price of anarchy} as the ratio of the system cost in the worst-case equilibrium and the optimal system cost.
$$\poa = \frac{\max_{\strat^{\rm EQ} \in {\rm EQ}} C(\strat^{\rm EQ})}{\min_{\Opt{\strat} \in \Strats}C(\Opt{\strat})} \geq 1,$$
where ${\rm EQ} \subseteq \Strats$ denotes the set of all\footnote{A joint strategy satisfying \eqref{eq:def_gen_eq} need not exist (i.e., ${\rm EQ} = \emptyset$). In Proposition~\ref{prop:exist}, we show existence with common-interest collaboration.} joint strategies satisfying \eqref{eq:def_gen_eq}.
By changing the collaborative decision-making structure, we alter the set ${\rm EQ}$ and attain different price of anarchy guarantees.

We will particularly focus on how the price of anarchy changes as the amount of collaboration increases.
To do so, throughout this work, we consider the specific collaborative structure in which each group of agents up to size $k$ can collaborate in choosing their group action, i.e., a group $\Gamma$ can select the group strategy $\strat_\Gamma \in \Strats_\Gamma = \prod_{i \in \Gamma} \Strats_i$ if they are in the collaboration set $\Combos_N^{[k]} = \bigcup_{\zeta = 1}^k \Combos_N^\zeta$, where $\Combos_N^k = \{\Gamma \subset N \mid |\Gamma|=k\}$ is the set of groups containing exactly $k$ agents.
We focus on this collaborative structure as it presents a symmetry among the agents that is useful in our worst-case analysis, and by varying $k$ between $1$ and $n$, we can range between fully distributed and fully centralized decision-making architectures.

With the particular collaborative structure $(\Combos_N^{[k]},\Phi)$, the states that satisfy \eqref{eq:def_gen_eq} are the \kstrong Nash equilibria of the common-interest game~\cite{aumann1959acceptable}.
We will denote the set of all such states by $\KSNE(\Phi) \subseteq S$.
Many works have studied the existence of \kstrong Nash equilibria \cite{nessah2014existence,harks2009strong} and provide methods to find them~\cite{holzman1997strong,gatti2017verification,clempner2020finding} in general and specific classes of games.
In this work, we address the quality of these states in congestion games with common-interest collaboration.
We will denote by $\spoa_k(G,\Phi)$ the \kstrong price of anarchy in the multi-agent system $G$ with common-interest objective $\Phi$.
We will seek to quantify this equilibrium performance metric in the class of atomic congestion games.

\subsection{Congestion Games}\label{subsec:congestion}
We consider the effects of collaborative decision-making in atomic congestion games.
Let $\mathcal{E}$ be a set of resources to be used by a set of agents $N$.
Each agent $i \in N$ can select a set of resources $\strat_i \subseteq \mathcal{E}$ from a constrained subset of the power set of resources $\Strats_i \subseteq 2^\mathcal{E}$.
As more agents share the use of a resource, its congestion and associated costs increase.
Let $|\strat|_e = |\{i \in N \mid e \in \strat_i\}|$ denote the number of agents utilizing resource $e$ in the joint strategy $\strat$ and 
let $\ell_e:\{0,\ldots,n\}\rightarrow\mathbb{R}_{\geq 0}$ denote the latency function that models the cost of congestion for resource $e$.
In the classic setting of local, selfish routing, an individual's cost is the total latency of the resources they use, i.e., $J_i(\strat) = \sum_{e \in \strat_i} \ell_e(|\strat|_e)$.
Alternatively, the system operator seeks to minimize the total congestion in the system; as such, we let 
$$C(\strat) = \sum_{e \in \mathcal{E}} |\strat|_e\ell_e(|\strat|_e),$$
define the system-level objective, which captures the total cost observed by all agents.
We will consider a congestion game $G = (N,\mathcal{E},S,\{\ell_e\}_{e \in \mathcal{E}})$ as the multi-agent system of focus for the remainder of this work.


In congestion games, the most predominantly studied agent decision-making model is that of selfish routing~\cite{Roughgarden2002}.
When each agent minimizes their own observed cost $J_i(\strat) = \sum_{e \in \strat_i} \ell_e(|\strat|_e)$, it is well known that the emergent equilibrium (Nash equilibrium) can be sub-optimal~\cite{Koutsoupias2009}.
In the case of affine congestion games, when agents minimize their own observed costs, the price of anarchy can be as large as 2.5~\cite{Christodoulou2005}.
This may lead one to believe that inefficiency is caused by the misalignment of users' self-interested preferences and the system-level objective; as such, one would expect that when agents are altruistic and seek to minimize the total congestion (i.e., $J_i(\strat)=C(\strat)$), the price of anarchy would decrease, but in fact, it increases to 3~\cite{caragiannisTaxesLinearAtomic2006,Paccagnan2021I}.
The best-designed objective for local decision-making can reduce the price of anarchy guarantee only as low as approximately 2.012~\cite{Paccagnan2021I}.
These findings indicate that a significant source of inefficiency is not solely caused by the misalignment of objectives but also by the local nature of decision-making.
When collaborative decision-making is introduced, new questions emerge, such as 1) How does the price of anarchy decrease with collaboration in each group decision-making framework? and 2) what level of collaboration is needed for altruistic behavior to outperform local, selfish behavior?
The tools provided in this work help to answer these questions; \cref{fig:spoa_grid} illustrates our findings in several classes of congestion games.


\section{Altruistic Collaboration}\label{sec:altr}
In this section, we study the effects of collaborative altruism on equilibrium efficiency; we focus this in the class of altruistic congestion games in which each group of agents $\Gamma \in \Combos \subset 2^N$ has the objective of minimizing the system cost function, i.e., $J_\Gamma(\strat) = C(\strat) = \sum_{e \in \mathcal{E}} \lvert\strat\rvert_e\ell_e(\lvert\strat\rvert_e)$.
With the particular collaborative structure $\Combos_N^{[k]}$ in which groups of up to $k$ agents can form collaborative groups, the equilibria of the system are states where no group $\Gamma \in \Combos_N^{[k]}$ can deviate their action to reduce the system objective.
In the context of common interest games, this is precisely the set of \kstrong Nash equilibria.
\begin{definition}\label{def:def_ksnash}
    In an altruistic congestion game $(G,C)$ a joint strategy $\KSNash{\strat} \in \Strats$ is a \kstrong Nash equilibrium if
    \begin{equation}\label{eq:def_ksnash}
        C(\KSNash{\strat}) \leq C(\strat_\Gamma^\prime,\KSNash{\strat}_{-\Gamma}),\quad\forall\Gamma\in\Combos_N^{[k]}.
    \end{equation}
\end{definition}
Recall $\KSNE(C)$ as the set of all such \kstrong Nash equilibria in the game $(G,C)$.
We highlight the explicit definition of \kstrong Nash equilibria in altruistic congestion games as (1) it aligns with the more general framework of collaborative decision-making introduced in \cref{subsec:prelim_collab}, and (2) it demonstrates the fact that \kstrong Nash equilibria are the local optima of the system cost function in the neighborhood of $k$-lateral deviations.
In the special case where $C$ is super- or sub-modular, we can think of these results as extending the existing literature on greedy approximation algorithms in sub-modular (or often super-modular) maximization~\cite{xu2022ResourceAwareDistributedSubmodular}.

The main contributions of this work are in understanding the quality of \kstrong Nash equilibria.
Before doing so, we show that such equilibria are guaranteed to exist.
\begin{proposition}\label{prop:exist}
For any altruistic congestion game $(G,C)$, for any $k \in [n]$, there exists a \kstrong Nash equilibrium.
\end{proposition}
\begin{proof}
The claim is shown by observing that the agent cost minimizers are \kstrong\ Nash equilibria for all $k \in [n]$.
Let $\hat{\strat} \in \argmin_{\strat \in \Strats}C(\strat)$ be a joint strategy that minimizes the system cost; by definition
$C(\Opt{\strat}) \leq C(\strat) ~\forall\strat \in \Strats.$
The set of $k$-lateral deviations from $\hat{\strat}$ is $\hat{\Strats} = \{(\strat_\Gamma^\prime,\hat{\strat}_{-\Gamma}) \mid \strat_\Gamma \in \Strats_\Gamma,\Gamma \in \Combos_N^{[k]}\}$.
As $\hat{\Strats} \subseteq \Strats$ for all $k \in [n]$, \eqref{eq:def_ksnash} holds from the optimality conditions.
\end{proof}

Though \kstrong equilibria exist, they need not be optimal in the system objective $C$.
In light of this, we consider the \kstrong price of anarchy as the ratio of how well a \kstrong Nash equilibrium approximates the optimal state.
Let,
$$\spoa_k(G,C) = \frac{\max_{\KSNash{\strat}\in\KSNE} C(\KSNash{\strat})}{\min_{\Opt{\strat} \in \Strats} C(\Opt{\strat})},$$
denote the \kstrong price of anarchy of an altruistic congestion game $(G,C)$.
As this ratio approaches 1, the quality of a collaborative equilibrium approaches the optimal.
Said differently, in the context of altruistic congestion games, as we increase the level of collaboration through the parameter $k$, agents reach equilibria with lower system costs.
By studying the \kstrong price of anarchy, we can quantify how going beyond local decision-making and considering paradigms between centralized and distributed can provide us with performance improvements without requiring full coordination.

To gain broader insights across many problem instances, we are particularly interested in understanding equilibrium efficiency across classes of congestion games.
Let $\Latset = \{\ell^1,\ldots,\ell^L\}$ denote a set of basis latency functions and $\gee_\Latset^n$ denote the set of congestion games with at most $n$ agents and resources that have latencies of the form $\ell_e = \sum_{j=1}^L \alpha_e^j \ell^j$.
We assume each basis latency function maps to the non-negative real numbers.
Clearly, any specific game $G$ can be found in the class $\gee^{|N|}_{\{\ell_e\}_{e \in \mathcal{E(G)}}}$; however, many congestion games of interest exist among several well-studied classes of games such as linear congestion games where $\Latset = \{1,x\}$~\cite{Caragiannis2010}, Polynomial congestion games where $\Latset = \{1,x,\ldots,x^D\}$ where $D$ is some highest order polynomial modeled~\cite{Aland2006,Bilo2016}, and exponential where $\Latset = \{e^{\alpha_1x},\ldots,e^{\alpha_mx}\}$ \cite{chien2009strong}.
We extend the definition of \kstrong price of anarchy to provide a bound over a class of problems, i.e.,
\begin{equation}
    \spoa_k(\gee_\Latset^n) = \max_{(G,C) \in \gee_\Latset^n} \spoa_k(G,C).
\end{equation}

In \cref{thm:alt} we provide tractable linear programs that provide lower and upper bounds on the \kstrong price of anarchy for any $\gee_\Latset^n$.
For notational convenience, let $c^j(x) = x\ell^j(x)$, and without loss of generality let $\Latset$ define a set of local cost functions $\{c^j\}_{j \in |\Latset|}$.
\begin{theorem}\label{thm:alt}
For the class of altruistic congestion games $\gee_\Latset^n$, the \kstrong price of anarchy satisfies
\begin{equation}\label{eq:thm_alt}
    \min_{\zeta \in [k]}\left\{1/P_\zeta^\star(n,\Latset)\right\} \geq \spoa_k(\gee_\Latset^n) \geq \max_{c \in \Latset}\left\{ 1/Q_k(n,c)\right\},
\end{equation}
where
\begin{align}
&P_\zeta^\star(n,\Latset) = \max_{\rho\geq\nu \geq 0} ~~~ \rho \nonumber\\
&{\rm s.t.} \hspace{8pt}0 \leq c(b\hs +\hs x)-\rho c(a\hs +\hs x) + \nonumber\\
& \hs\hs\hs\left( \hs \hs \hs \hs \binom{n}{\zeta} c(a\hs +\hs x) \hs -\hs  \hs \hs \hs \sum_{\substack{0\leq\psi\leq a\\ 0\leq\omega\leq b\\ \psi+\omega \leq \zeta }}\hs  \hs {a\choose \psi} \hs {b \choose \omega} \hs \binom{n\hs - \hs a\hs - \hs b}{\zeta \hs - \hs \psi \hs - \hs \omega} c (a\hs +\hs x\hs +\hs \beta\hs -\hs \alpha)  \hs \hs \hs \right)\nonumber\\
&\hspace{120pt}\forall (a,x,b) \in \mathcal{I},~c \in \Latset.\tag{P$\zeta$}\label{opt:ub}
\end{align}
and
\begin{align}
Q_k^\star(n,c) =& \min_{\theta \in \mathbb{R}^{|\mathcal{I}|}_{\geq 0}} \quad \sum_{a,x,b}c(b+x)\theta(a,x,b)\nonumber\\
~~{\rm s.t.} & \sum_{a,x,b}c(a+x)\theta(a,x,b)=1\nonumber\\
 &0 \geq \sum_{a,x,b}\theta(a,x,b)\hs \Bigg(\hs \hs \binom{n}{\zeta}c(a \hs + \hs x)\nonumber\\
 & -\sum_{\substack{0\leq\psi\leq a\\ 0\leq\omega\leq b\\ \psi+\omega \leq \zeta }} \hs \hs {a\choose \psi} \hs {b \choose \omega} \hs\binom{n\hs \hspace{-0.5pt} - \hspace{-0.5pt} \hs a\hs \hspace{-0.5pt} -\hspace{-0.5pt}  \hs b}{\zeta \hs \hspace{-0.5pt} - \hspace{-0.5pt} \hs \psi \hs \hspace{-0.5pt} - \hspace{-0.5pt} \hs \omega} c (a\hs +\hs x\hs +\hs \omega\hs -\hs \psi) \hs\hs  \Bigg) \nonumber \\
 & \hspace{100pt}\forall\zeta \in \{1,\ldots,k\}\tag{Q$_k$}\label{opt:lb}
\end{align}
\end{theorem}

In total, we provide $k+|\Latset|$ linear programs to solve with $\mathcal{O}(n^3|\Latset|)$ constraints and decision variables.
The first set of $k$ programs \eqref{opt:ub} each provide an upper bound on the \kstrong price of anarchy, where \eqref{eq:thm_alt} simply states that the lowest such bound will provide the best guarantee on $\spoa_k(\gee_\Latset^n)$.
Similarly, the linear programs \eqref{opt:lb} constructs $|\Latset|$ examples with large \kstrong price of anarchy.
In \cref{fig:spoa_grid}, we demonstrate the results of these plots for four classes of congestion games: affine, quadratic, quartic, and asymptotic.
In each plot, only one line is shown as the two bounds coincide for each setting.
This indicates that our approach provides extremely relevant bounds on the \kstrong price of anarchy, which can be used to understand the opportunities to reduce inefficiency by partially coordinating agent behavior.

\noindent\emph{Proof of Theorem~\ref{thm:alt}:}

We deconstruct the proof into two parts: the proof of the upper bound and the proof of the lower bound:

\noindent\textbf{Proof of upper bound:} To prove the upper bound on $\spoa_k(\gee_\Latset^n)$, observe that if there exist vectors $\lambda,\mu \in \mathbb{R}^k_{\geq 0}$ such that for any pair of joint strategies $\strat,\strat^\prime \in \Strats$,
\begin{equation}\label{eq:smooth_def}
\frac{1}{\binom{n}{\zeta}} \sum_{\Gamma \in \Combos_N^{\zeta}} C(\strat_\Gamma^\prime,\strat_{-\Gamma}) \leq \lambda_\zeta C(\strat^\prime) - \mu_{\zeta}C(\strat),~\forall \zeta \in [k],
\end{equation}
then the following upper bounds hold on the system cost of a \kstrong Nash equilibrium:
\begin{subequations}
\begin{align}
C(\KSNash{\strat}) &= \frac{1}{\binom{n}{\zeta}} \sum_{\Gamma \in \Combos_N^\zeta} C(\KSNash{\strat})\label{eq:smooth_proof_a}\\
&\leq \frac{1}{\binom{n}{\zeta}} \sum_{\Gamma \in \Combos_N^\zeta} C(\Opt{\strat}_\Gamma, \KSNash{\strat}_{-\Gamma})\label{eq:smooth_proof_b}\\
&\leq \lambda_\zeta C(\Opt{\strat}) - \mu_{\zeta}C(\KSNash{\strat}).\label{eq:smooth_proof_c}
\end{align}
\end{subequations}
Rearranging terms gives that any altruistic congestion game satisfying \eqref{eq:smooth_def} has
\begin{equation}
    \frac{C(\KSNash{\strat})}{C(\Opt{\strat})} \leq \frac{\lambda_\zeta}{1+\mu_\zeta} \forall \zeta \in [k],
\end{equation}
for all $\KSNash{\strat} \in \KSNE$, i.e., $\spoa_k(G,C)$ is upper bounded by a simple function of $\lambda$ and $\mu$.
With this, our search for an upper bound reduces to a search for vectors $\lambda$ and $\mu$ that satisfy \eqref{eq:smooth_def} for each altruistic congestion game $(G,C) \in \gee_\Latset^n$.

The final step of the proof of the lower bound is to construct tractable linear programs whose solutions proved a $\lambda$ and $\mu$ that satisfy \eqref{eq:smooth_def}.
Consider an altruistic congestion game $(G,C)$ and any two joint strategies $\strat,\strat^\prime \in \Strats$.
To each resource $e \in \mathbb{E}$, we assign the label $(a_e,x_e,b_e)$ defined by 
\begin{subequations}\label{eq:parameterization}
\begin{align}
    a_e &= \lvert\{ i \in N \mid e \in \strat_i \setminus \strat^\prime_i\}\rvert \\
    x_e &= \lvert\{ i \in N \mid e \in \strat_i \cap \strat^\prime_i\}\}\rvert \\
    b_e &= \lvert\{ i \in N \mid e \in \strat^\prime_i \setminus \strat_i\}\}\rvert.
\end{align}
\end{subequations}
That is, $a_e$ is the number of players who utilize resource $e$ in only joint strategy $\strat$, $b_e$ is the number of players who utilize resource $e$ in only joint strategy $\strat^\prime$, and $x_e$ is the number of players who utilize resource $e$ in both $\strat$ and $\strat^\prime$.
Let $\mathcal{I} := \{(a,x,b)\in \mathbb{N}_{\geq 0}^3 \mid 1 \leq a+x+b \leq n\}$ denote the set of all possible labels in a game that has $n$ players.
Further, let $\mathcal{E}_{(a,x,b)} = \{e \in \mathcal{E} \mid (a_e,x_e,b_e)=(a,x,b)\}$ denote the set of resources with label $(a,x,b)$, and let $\theta(a,x,b,j) = \sum_{e \in \mathcal{E}_{a,x,b}} \alpha_e^j$ denote the aggregate cost scaling factor for the basis latency function $c^j$ on resources with label $(a,x,b)$.
$\theta$ is a vector in $\mathbb{R}_{\geq 0}^{|\Latset|\cdot|\mathcal{I}|}$ and will be used to rewrite \eqref{eq:smooth_def}.

Notice that the term $C(\strat^\prime) = \sum_{e \in \mathcal{E}}\sum_{j \in |\Latset|} \alpha_e^j c^j(|\strat|_e)$ depends only on the number of players utilizing resource $e$.
Using the aforementioned parameterization, we can write $|\strat^\prime|_e = b_e+x_e$ as the number of agents using resource $e$ in joint action $\strat^\prime$.
Rearranging terms allows us to rewrite the optimal total congestion further as
\begin{align*}
    C(\strat^\prime) &= \sum_{e \in \mathcal{E}}\sum_{j \in |\Latset|} \alpha_e^j c^j(|\strat|_e)\\
        &=\sum_{a,x,b,j}\sum_{e \in \mathcal{E}_{(a,x,b)}} \alpha_e^j c^j(b+x)\\
        &=\sum_{a,x,b,j}c^j(b+x)\sum_{e \in \mathcal{E}_{(a,x,b)}} \alpha_e^j\\
        &=\sum_{a,x,b,j}c^j(b+x)\theta(a,x,b,j).
\end{align*}
when omitted, it is assumed that a sum over $a,x,b,j$ is over all labels in $\mathcal{I}$ and basis cost function indicies $j \in [|\Latset|]$.
We can similarly rewrite the system cost $C(\strat)$ using the same parameterization as $\sum_{a,x,b,j}c^j(a+x)\theta(a,x,b,j)$.

We similarly transcribe the final term in \eqref{eq:smooth_def} $\sum_{\Gamma \in \Combos_N^\zeta} C(\strat_\Gamma^\prime,\strat_{-\Gamma})$ using the same parameterization.
\begin{align*}
&\sum_{\Gamma \in \Combos_N^\zeta}  C(\strat^\prime_\Gamma,\strat_{-\Gamma})\\
 &= \sum_{\Gamma\in\Combos_N^\zeta} \sum_{a,x,b,j} \sum_{e \in \mathcal{E}_{(a,x,b)}} \alpha_e^j c^j(|\strat^\prime_\Gamma,\strat_{-\Gamma}|_e)\\
&= \sum_{a,x,b,j} \sum_{e \in \mathcal{E}_{(a,x,b)}} \alpha_e^j \sum_{\Gamma\in\Combos_N^\zeta} c^j(|\strat^\prime_\Gamma,\strat_{-\Gamma}|_e)\\
&=  \hs \hs  \hs \sum_{a,x,b,j}  \hs  \hs \hs  \theta(a,x,b,j) \hs \hs  \hs  \hs \sum_{\substack{0 \leq \psi \leq e\\ 0\leq \omega \leq b \\ \psi+\omega \leq \zeta}} \hs  \hspace{-2mm} {a\choose \psi}  \hs {b \choose \omega} \hs \binom{n\hs - \hs a\hs - \hs b}{\zeta \hs - \hs \psi \hs - \hs \omega}c^j(a  \hs  \hs + \hs  \hs  x  \hs  \hs + \hs  \hs  \omega \hs  \hs - \hs \hs  \psi)
\end{align*}
where the set of coalitions $\Combos_N^\zeta$ was partitioned according to the action profile of the agents in each coalition.
We let $\psi$ denote the number of agents in $\Gamma$ that utilize resource $e$ only in joint action $\strat$ and $\omega$ the number of agents in $\Gamma$ that utilize $e$ only in joint action $\strat^\prime$.
By simple counting arguments, there are exactly ${a\choose \psi} {b \choose \omega}\binom{n -  a -  b}{\zeta  - \psi -  \omega}$ coalitions grouped with the same $\psi$ and $\omega$.
This decomposition is possible as the number of agents utilizing resource $e$ after a group $\Gamma$ deviates is precisely $a+x+\omega-\psi$.

The inequality \eqref{eq:smooth_def} now becomes
\begin{align*}
    &\frac{1}{\binom{n}{\zeta}}\hs\hs \hs \sum_{a,x,b,j}\hs\hs\hspace{-2pt} \theta(a,x,b,j) \hs\hs\hs\hs\hs\sum_{\substack{0 \leq \psi \leq e\\ 0\leq \omega \leq b \\ \psi+\omega \leq \zeta}} \hspace{-2mm}\hspace{-2pt} {a\choose \psi} \hs\hs {b \choose \omega}\hs\hs\binom{n\hs - \hs a\hs - \hs b}{\zeta \hs - \hs \psi \hs - \hs \omega}c^j(a\hs\hs +\hs\hs  x\hs\hs  +\hs\hs \omega \hs\hs - \hs\hs\psi) \\
    &\leq \sum_{a,x,b,j}\hs \theta(a,x,b,j) \left( \lambda_\zeta c^j(b+x) \hs - \hs \mu_\zeta c^j(a+x)\right).
\end{align*}
We assume that each latency function is a non-negative linear combination of the basis latency functions, as such $\theta(a,x,b,j) \geq 0$ for each $(a,x,b) \in \mathcal{I}$ and $j \in [|\Latset|]$.
Rather than satisfy this one inequality, it is sufficient to satisfy the $|\Latset|\cdot|\mathcal{I}|$ inequalities
\begin{multline}\label{eq:param_constraint}
    \frac{1}{\binom{n}{\zeta}} \sum_{\substack{0 \leq \psi \leq e\\ 0\leq \omega \leq b \\ \psi+\omega \leq \zeta}} \hspace{-2mm} {a\choose \psi} {b \choose \omega}\binom{n\hs - \hs a\hs - \hs b}{\zeta \hs - \hs \psi \hs - \hs \omega}c(a + x + \omega - \psi) \\ \leq \lambda_\zeta c(b+x) - \mu_\zeta c(a+x), \forall (a,x,b) \in \mathcal{I},~ c \in \Latset.
\end{multline}

To find parameters $\lambda_\zeta$ and $\mu_\zeta$ that provide the most informative \kstrong price of anarchy upper bound, we formulate the following optimization problem:
\begin{align}
\min_{\lambda_\zeta,\mu_\zeta \geq 0} \quad & \frac{\lambda_\zeta}{1+\mu_\zeta}\tag{P1$\zeta$}\label{opt:smooth_comp}\\
{\rm s.t.} \quad &~~ \eqref{eq:param_constraint}\nonumber
\end{align}
Finally, we transform \eqref{opt:smooth_comp} by substituting new decision variables $\rho = (1+\mu_\zeta)/\lambda_\zeta$ and $\nu = 1/\left(\binom{n}{\zeta}\lambda_\zeta\right) \geq 0$.
The new objective becomes $1/\rho$.
Except in degenerate cases where the \kstrong price of anarchy is unbounded, $\rho > 0$; we can thus invert the objective and change the minimization to a maximization, giving \eqref{opt:lb}.

\noindent\textbf{Proof of upper bound:}
For each $c \in \Latset$ we will algorithmically construct an example which provides large \kstrong price of anarchy\footnote{further analysis that is omitted from this work shows that these examples give tight bounds for the classes of games with a single basis latency function.}.
Let $\hat{\gee}_c^n$ denote the class of congestion games with one basis latency function $c=x\ell(x)$ and each player possessing exactly two actions $\Strats_i = \{\KSNash{\strat}_i,\Opt{\strat_i}\}$.
The cost function of a resource $e \in \mathcal{E}$ in these games is parameterized only by a single scalar $\alpha_e$.
Consider the problem
\begin{align}\label{eq:OG_prog_construct}
\max_{G \in \hat{\gee}_c^n} \frac{\max_{\KSNash{\strat} \in \KSNE}C(\KSNash{\strat})}{\min_{\Opt{\strat}\in\Strats}C(\Opt{\strat})}
\end{align}
whose solution is the upper bound on the \kstrong\ price of anarchy for $\hat{\gee}_c^n$.
We will make a series of reductions and generalizations to this program.

First, by leveraging the structure in $\hat{\gee}_c^n$ that each player has two actions, we can rewrite \eqref{eq:OG_prog_construct} as
\begin{align}\label{eq:OG_prog_construct_1}
\max_{G \in \hat{\gee}_c^n}~~~~&~~ \frac{C(\KSNash{\strat})}{C(\Opt{\strat})}\nonumber\\
{\rm s.t.}~~&~~ C(\KSNash{\strat}) \leq C(\Opt{\strat}_\Gamma,\KSNash{\strat}_{-\Gamma}),~\forall \Gamma \in \Combos_N^{[k]}.
\end{align}
Next, we generalize the constraint set.
For each $\zeta \in [k]$ we add the constraints associated with deviating groups of size $\zeta$ into a single constraint,
resulting in the problem 
\begin{align}\label{eq:OG_prog_construct_2}
\max_{G \in \hat{\gee}_c^n}~~~~&~~ \frac{C(\KSNash{\strat})}{C(\Opt{\strat})}\nonumber\\
{\rm s.t.}~~&~~ \binom{n}{\zeta}C(\KSNash{\strat}) \leq \sum_{\Gamma \in \Combos_N^\zeta} C(\strat_\Gamma^\prime,\KSNash{\strat}_{-\Gamma}),~\forall \zeta\in [k].
\end{align}
The feasible set of \eqref{eq:OG_prog_construct_2} subsumes that of \eqref{eq:OG_prog_construct_1}, and thus \eqref{eq:OG_prog_construct_2} provides an upper bound on \eqref{eq:OG_prog_construct_1} and the original problem \eqref{eq:OG_prog_construct}.

Finally, we alter \eqref{eq:OG_prog_construct_2} by the following: add the constraint that $C(\KSNash{\strat})=1$, which does not alter the value of the problem as resource values $\alpha_e$ can be scaled without altering the ratio, and invert the objective while turning the maximization to a minimization, which also does not alter the value of the problem as $C(\Opt{s}) \geq 0$.
These changes give the program
\begin{align}\label{eq:OG_prog_construct_3}
\min_{G \in \hat{\gee}_n}~~~~&~~ C(\Opt{\strat})\nonumber\\
{\rm s.t.}~~&~~ \binom{n}{\zeta}C(\KSNash{\strat}) \leq \sum_{\Gamma \in \Combos_N^\zeta} C(\Opt{\strat}_\Gamma,\KSNash{\strat}_{-\Gamma}),\forall \zeta\in [k],\nonumber\\
&~~C(\KSNash{\strat}) = 1.
\end{align}
In the next step of this proof, we will parameterize the terms in \eqref{eq:OG_prog_construct_3} using the parameterization from the first part of this proof.

By design, for each congestion game $G \in \hat{\gee}_c^n$ each player $i \in N$ has exactly two actions $\KSNash{\strat}_i \subseteq \mathcal{E}$ and $\Opt{\strat}_i \subseteq \mathcal{E}$.
To each resource $e \in \mathbb{E}$, we assign the label $(a_e,x_e,b_e)$ defined as in \eqref{eq:parameterization} where $\strat = \KSNash{\strat}$ and $\strat^\prime = \Opt{\strat}$.
Let $\mathcal{I} := \{(a,x,b)\in \mathbb{N}_{\geq 0}^3 \mid 1 \leq a+x+b \leq n\}$ denote the set of all possible labels in $\hat{\gee}_c^n$.
Recall the parameter $\theta(a,x,b) = \sum_{e \in \mathcal{E}_{a,x,b}} \alpha_e$ that is a vector in $\mathbb{R}_{\geq 0}^{|\mathcal{I}|}$.
Using the parameterized terms of $C(\KSNash{\strat})$, $C(\Opt{\strat})$, and $\sum_{\Gamma \in \Combos_N^\zeta} C(\Opt{\strat}_\Gamma,\KSNash{\strat}_{-\Gamma})$ shown in the first part of this proof, \eqref{eq:OG_prog_construct_3} can be rewritten as \eqref{opt:lb}.
By optimizing over $\theta$, the program searches for the worst-case price of anarchy over grouped resource values.
We point out that \eqref{opt:lb} is a linear program in $\theta$ with $k$ linear inequality constraints and one linear equality constraint.
By the transformation steps in the first part of this proof, \eqref{opt:lb} has been show to provide an upper bound on the \kstrong\ price of anarchy in $\hat{\gee}_c^n$; in the following step, we will provide a construction which shows this bound is tight.

Consider the following congestion game: for each label $(a,x,b) \in \mathcal{I}$ and permutation of the $n$ players $\sigma \in {\Sigma}_n$, define a ring of $n$ resources.
Total, there are $nn!|\mathcal{I}|$ resources.
Let $e_{i,j}^{(a,x,b)}$ denote the resource with label $(a,x,b)$ at position $i$ in the $j$th ring.
In \cref{fig:rings}, we illustrate the first indices of the $n!$ rings associated with the label $(a,x,b)=(2,1,1)$.

In the constructed congestion game, let $$\KSNash{\strat}_i = \bigcup_{\substack{1\leq j \leq n! \\ (a,x,b) \in \mathcal{I}}} e_{\sigma(i),j}^{(a,x,b)} \cup \ldots \cup e_{(\sigma(i) + a + x -1)\% n,j}^{(a,x,b)}$$ that is, in each ring $j$ with label $(a,x,b)$, at position $\sigma(i)$ player $i$ uses resource $e_{\sigma(i),j}^{(a,x,b)}$ as well as the next $a+x-1$ resources in the ring (for a total of $a+x$).
If these indices surpass the number of resources in the ring (i.e., $\sigma(i)+a+x-1 > n$, then the incrementing restarts at the beginning of the ring (hence the modulus operator $\%$).
Similarly, for the optimal actions, let $$\Opt{\strat}_i = \bigcup_{\substack{1\leq j \leq n! \\ (a,x,b) \in \mathcal{I}}} e_{(\sigma(i)+a)\% n,j}^{(a,x,b)} \cup \ldots \cup e_{(\sigma(i) + a + x +b -1)\% n,j}^{(a,x,b)}$$
with the same repeating pattern around the ring.
Constructing the action sets in this way enforces that each resource in a ring with label $(a,x,b)$ is utilized by $a+x$ players in the action $\KSNash{\strat}$ and $b+x$ players in the action $\Opt{\strat}$.

To finish constructing the example, let each resource $e \in \mathcal{E}_{(a,x,b)}$ have a cost scaling factor $\alpha_e = \theta(a,x,b)$, where we can select any $\theta \in \mathbb{R}_{\geq 0}^{|I|}$.
We will call this construction $G_\theta$.

\begin{figure}
    \centering
    \vspace{10pt}
    \includegraphics[width=0.4625\textwidth]{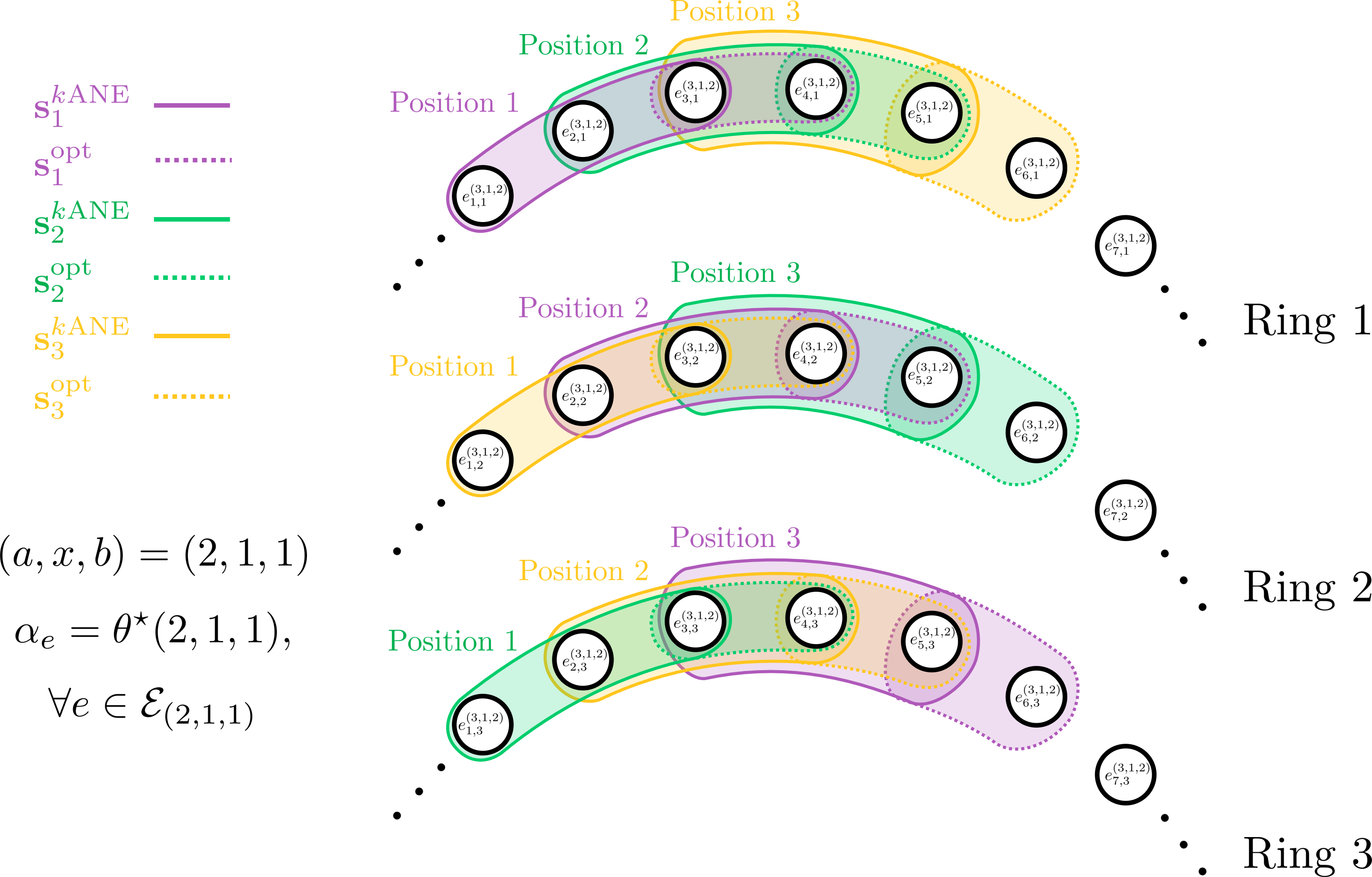}
    \caption{Game construction for worst-case \kstrong\ price of anarchy. Three of the $n$ players' action sets are shown (color-coded in yellow, green, and purple, respectively) on three of $n!$ rings for the label $(a,x,b)=(2,1,1)$. A ring has $n$ positions, one for each player. For a label $(a,x,b)$, we generate $n!$ rings for all the orderings of players over positions. This is repeated for each label. Players still only have two actions, but each action covers resources from each ring. The value of a resource with label $(a,x,b)$ is equal to a value of $\theta^\star(a,x,b)$, which we can set as equal to a solution to \eqref{opt:lb}.}
    \label{fig:rings}
    \vspace{-10pt}
\end{figure}

We will now investigate the \kstrong\ price of anarchy in the construction $G_\theta$ with latency function $\ell$ and local cost function $c(x) = x\ell(x)$.
In the congestion game $(G_\theta,C)$, the total latency in the action $\KSNash{\strat}$ is
\begin{equation}\label{eq:C_ksne_welf}
C(\KSNash{\strat}) = \sum_{a,x,b}n n!\theta(a,x,b)c(a+x).
\end{equation}
Similarly, joint strategy $\Opt{\strat}$ is
\begin{equation}\label{eq:C_opt_welf}
	C(\Opt{\strat}) = \sum_{a,x,b} n n!\theta(a,x,b)c(b+x).
\end{equation}
The system cost of a group $\Gamma$ deviating their action to $\Opt{\strat}_\Gamma$ from the joint strategy $\KSNash{\strat}$ is
\begin{align}
&C(\Opt{\strat}_\Gamma,\KSNash{\strat}_{-\Gamma}) = \sum_{a,x,b}\sum_{j=1}^{n!} \sum_{i=1}^n \theta(a,x,b)c(|\Opt{\strat}_\Gamma,\KSNash{\strat}_{-\Gamma}|_e)\nonumber\\
&\hs \hs =\hs \sum_{a,x,b}\hs\theta(a,x,b)\hs \hs \sum_{\substack{0\leq\psi\leq a\\ 0\leq\omega\leq b\\ \psi+\omega \leq \zeta }} \hs \hs \hs \hs nn! \hs  {a\choose \psi} \hs \hs {b \choose \omega} \hs \hs \binom{n\hs - \hs a\hs - \hs b}{\zeta \hs - \hs \psi \hs - \hs \omega} c (a\hs \hs +\hs \hs x\hs \hs +\hs \hs \omega\hs \hs -\hs \psi)\label{eq:C_gamma_dev}
\end{align}
where we let $e$ be the shorthand for $e_{i,j}^{(a,x,b)}$.
The second equality holds by defining $\psi$ and $\omega$ as the number of players in $\Gamma$ who invested in resource $e$ exclusively in their action $\KSNash{\strat}$ or $\Opt{\strat}$ respectively.
By counting arguments, there are exactly $\binom{a}{\psi}\binom{b}{\omega}\binom{n-a-b}{\zeta-\psi-\omega}$ positions for the players in $\Gamma$ which yield the profile $(\psi,\omega)$ for a resource at some fixed position in the ring, there are $\zeta!$ ways to order the players in $\Gamma$, $(n-\zeta)!$ ways to order the players not in $\Gamma$, and $n$ resource in each ring.
Due to the symmetry of the game, the system cost after a deviation is the same for any group $\Gamma \in \Combos_N^k$.

Now, we provide a condition on $\theta$ such that $\KSNash{\strat}$ is a \kstrong\ Nash equilibrium: we constrain that $\Phi(\KSNash{\strat})$ must be no more than \eqref{eq:C_gamma_dev},
\begin{multline}
    \sum_{a,x,b} n n!\theta(a,x,b)c(a+x) \leq\\ \sum_{a,x,b} \theta(a,x,b)  \hs\sum_{\substack{0\leq\psi\leq a\\ 0\leq\omega\leq b\\ \psi+\omega \leq \zeta }} \hs \hs nn!  {a\choose \psi} {b \choose \omega}\binom{n\hs - \hs a\hs - \hs b}{\zeta \hs - \hs \psi \hs - \hs \omega} c (a\hs +\hs x\hs +\hs \omega\hs -\hs \psi)\\\forall \zeta \in [k].
\end{multline}
\sloppy Canceling $nn!$ gives the $k$ linear inequality constraints of \eqref{opt:lb}, i.e., $\theta$ lives in the same feasible set as that of \eqref{opt:lb}.

In a congestion game $(G_\theta,C)$, the system cost of $\Opt{\strat}$ is $\sum_{a,x,b} nn! \theta(a,x,b) c(b+x)$.
As this system cost upper bounds the optimal system cost, and $C(\KSNash{\strat})$ lower bounds the worst-case equilibrium cost when $\theta$ in the feasible set,
\begin{multline*}
    \frac{C(\KSNash{\strat})}{C(\Opt{\strat})} = \frac{\sum_{a,x,b}nn!\theta(a,x,b)c(a+x)}{\sum_{a,x,b}nn!\theta(a,x,b)c(b+x)}\\
    = \frac{1}{\sum_{a,x,b}\theta(a,x,b)c(b+x)} \leq \spoa_k(G_\theta,C)
\end{multline*}
where the second equality holds from canceling $nn!$ and the normalizing constraint $\sum_{a,x,b}\theta(a,x,b)c(a+x) = 1$.
With this, we have the following bounds
$$\frac{1}{Q_k^\star(n,c)}\geq \spoa_k(G_\theta,C) \geq \frac{1}{\sum_{a,x,b}\theta(a,x,b)c(b+x)},$$
for all feasible $\theta$.
The feasible set of \eqref{opt:lb} is non-empty from the existence of equilibria shown in \cref{prop:exist}.
When \eqref{opt:lb} is non-zero, selecting $\theta = \theta^\star$ as a solution of \eqref{opt:lb} shows that $1/Q_k^\star(n,c)$ must be tight over $\gee_n$.
If \eqref{opt:lb} is zero, then the \kstrong\ price of anarchy is unbounded.
\hfill$\qed$

\section*{Conclusion}
In this work, we studied the effect of collaborative decision-making on multi-agent system inefficiency by bounding the price of anarchy with respect to $k$-strong Nash equilibria.
Two linear programs were provided that give bounds on the price of anarchy in settings where groups of agents are altruistic in optimizing the system objective.
Using these tools, we are able to gain insights into the degree to which decentralized decision-making causes system inefficiency.
Future work will study less restrictive coalition structures, consider additional forms of collaborations between agents, and consider the transient effects of coalitional decision-making rather than just the equilibrium.


\bibliographystyle{IEEEtran} 
\bibliography{../../../My_Library}

\clearpage



\end{document}